\documentclass[a4paper,UKenglish,svgnames]{lipics-v2019}

\newcommand{\longversion}[1]{#1}
\newcommand{\shortversion}[1]{}

\usepackage{amssymb,amsmath,amsthm}
\usepackage[usenames,svgnames,dvipsnames]{xcolor}
\usepackage[poorman]{cleveref}
\usepackage{makecell}
\usepackage{multirow}
\usepackage{mathtools}
\usepackage{parskip}
\makeatletter
\def\thm@space@setup{%
  \thm@preskip=\parskip \thm@postskip=0pt
}
\makeatother
\hypersetup{
  colorlinks = true,
  linkbordercolor = {white},
  linkcolor = {IndianRed},
  citecolor = {DarkSeaGreen},
  urlcolor = {MediumVioletRed}
}

\newenvironment{proof_sketch}{\noindent{\em\sc Proof sketch:}}{ \hfill $\square$\\ }

\usepackage{enumitem}

\longversion{
\usepackage{charter}
\usepackage{amsthm}
\usepackage{eulervm}
}

\shortversion{
\setlength{\textfloatsep}{0.5cm}
\setlist[itemize]{leftmargin=*}
\setlength{\textfloatsep}{1pt}
\setlength{\intextsep}{1pt}
\setlength{\floatsep}{1pt}

\usepackage{titlesec}
}

\usepackage{url}
\usepackage[font=small,skip=0pt]{caption}

\newcommand{\YES}{{\sc Yes}\xspace}
\newcommand{\NO}{{\sc No}\xspace}
\newcommand{\OA}{{\sc Optimal Attack}\xspace}
\newcommand{\OD}{{\sc Optimal Defense}\xspace}

\newcommand{\SC}{{\sc Set Cover}\xspace}
\newcommand{\HS}{{\sc Hitting Set}\xspace}
\newcommand{\CL}{{\sc Clique}\xspace}
\newcommand{\kSum}{$k$-{\sc Sum}\xspace}
\newcommand{\gone}{{\sc greedy 1}\xspace}
\newcommand{\gtwo}{{\sc greedy 2}\xspace}

\newcommand{\caveat}{\ensuremath{\mathsf{CoNP\subseteq NP/Poly}}\xspace}

\usepackage{xspace}
\newcommand{\NP}{\ensuremath{\mathsf{NP}}\xspace}
\newcommand{\NPC}{\ensuremath{\mathsf{NP}}-complete\xspace}
\newcommand{\NPH}{\ensuremath{\mathsf{NP}}-hard\xspace}
\newcommand{\coNPH}{\ensuremath{\mathsf{coNP}}-hard\xspace}
\newcommand{\WTC}{\ensuremath{\mathsf{W[2]}}-complete\xspace}
\newcommand{\WT}{\ensuremath{\mathsf{W[2]}}\xspace}
\newcommand{\WO}{\ensuremath{\mathsf{W[1]}}\xspace}
\newcommand{\WTH}{\ensuremath{\mathsf{W[2]}}-hard\xspace}
\newcommand{\WOH}{\ensuremath{\mathsf{W[1]}}-hard\xspace}
\newcommand{\FPT}{\ensuremath{\mathsf{FPT}}\xspace}
\newcommand{\el}{\ensuremath{\ell}\xspace}
\newcommand{\suc}{\ensuremath{\succ}\xspace}

\newcommand{\CC}{\ensuremath{\mathcal C}\xspace}
\newcommand{\DD}{\ensuremath{\mathcal D}\xspace}
\newcommand{\EE}{\ensuremath{\mathcal E}\xspace}

\newcommand{\GG}{\ensuremath{\mathcal G}\xspace}
\newcommand{\HH}{\ensuremath{\mathcal H}\xspace}
\newcommand{\II}{\ensuremath{\mathcal I}\xspace}

\newcommand{\LL}{\ensuremath{\mathcal L}\xspace}

\newcommand{\OO}{\ensuremath{\mathcal O}\xspace}
\newcommand{\PP}{\ensuremath{\mathcal P}\xspace}
\newcommand{\QQ}{\ensuremath{\mathcal Q}\xspace}

\renewcommand{\SS}{\ensuremath{\mathcal S}\xspace}

\newcommand{\UU}{\ensuremath{\mathcal U}\xspace}
\newcommand{\VV}{\ensuremath{\mathcal V}\xspace}
\newcommand{\WW}{\ensuremath{\mathcal W}\xspace}

\newcommand{\ZZ}{\ensuremath{\mathcal Z}\xspace}

\usepackage{nicefrac}

\newcommand{\nfrac}{\nicefrac}

\newtheorem{observation}{\bf Observation}

\newcommand{\eps}{\varepsilon}
\renewcommand{\epsilon}{\eps}

\newcommand{\ignore}[1]{}

\newcommand{\pr}{\ensuremath{\prime}\xspace}

\renewcommand{\leq}{\leqslant}

\renewcommand{\ge}{\geqslant}
\renewcommand{\le}{\leqslant}

\usepackage{cleveref}

\crefname{theorem}{Theorem}{Theorems}
\crefname{observation}{Observation}{Observations}
\crefname{lemma}{Lemma}{Lemmas}
\crefname{corollary}{Corollary}{Corollaries}
\crefname{proposition}{Proposition}{Propositions}
\crefname{definition}{Definition}{Definitions}
\crefname{claim}{Claim}{Claims}
\crefname{reductionrule}{Reduction rule}{Reduction rules}

\newcommand{\WONE}{\textrm{\textup{W[1]}}}
\newcommand{\WTWO}{\textrm{\textup{W[2]}}}
\newcommand{\XP}{\textrm{\textup{XP}}}
\newcommand{\WP}{\textrm{\textup{W[P]}}}

\title{A Parameterized Perspective on Protecting Elections}

\author{Palash Dey}{Indian Institute of Technology, Kharagpur, India}{palash.dey@cse.iitkgp.ac.in}{}{}
\author{Neeldhara Misra}{Indian Institute of Technology, Gandhinagar, India}{neeldhara.m@iitgn.ac.in}{}{}
\author{Swaprava Nath}{Indian Institute of Technology, Kanpur, India}{swaprava@iitk.ac.in}{}{}
\author{Garima Shakya}{Indian Institute of Technology, Kanpur, India}{garima@cse.iitk.ac.in}{}{}

\authorrunning{P. Dey, N. Misra, S. Nath, and G. Shakya }


\hideLIPIcs{}
\begin{document}

\maketitle

\begin{abstract}
We study the parameterized complexity of the optimal defense and optimal attack problems in voting. In both the problems, the input is a set of voter groups (every voter group is a set of votes) and two integers $k_a$ and $k_d$ corresponding to respectively the number of voter groups the attacker can attack and the number of voter groups the defender can defend. A voter group gets removed from the election if it is attacked but not defended. In the optimal defense problem, we want to know if it is possible for the defender to commit to a strategy of defending at most $k_d$ voter groups such that, no matter which $k_a$ voter groups the attacker attacks, the outcome of the election does not change. In the optimal attack problem, we want to know if it is possible for the attacker to commit to a strategy of attacking $k_a$ voter groups such that, no matter which $k_d$ voter groups the defender defends, the outcome of the election is always different from the original (without any attack) one. 
We show that both the optimal defense problem and the optimal attack problem are computationally intractable for every scoring rule and the Condorcet voting rule even when we have only $3$ candidates. We also show that the optimal defense problem for every scoring rule and the Condorcet voting rule is \WTH for both the parameters $k_a$ and $k_d$, while it admits a fixed parameter tractable algorithm parameterized by the combined parameter $(k_a,k_d)$. The optimal attack problem for every scoring rule and the Condorcet voting rule turns out to be much harder -- it is \WOH even for the combined parameter $(k_a,k_d)$. We propose two greedy algorithms for the \OD problem and empirically show that they perform effectively on reasonable voting profiles.

\end{abstract}

\section{Introduction}

The problem of election control asks if it is possible for an external agent, usually with a fixed set of resources, to influence the outcome of the election by altering its structure in some limited way. There are several specific manifestations of this problem: for instance, one may ask if it is possible to change the winner by deleting $k$ voter groups, presumably by destroying ballot boxes or rigging electronically submitted votes. Indeed, several cases of violence at the ballot boxes have been placed on record~\cite{Bhattacharjya,RT}, and in 2010, Halderman and his students exposed serious vulnerabilities in the electronic voting systems that are in widespread use in several states~\cite{Halderman}. A substantial amount of the debates around the recently concluded presidential elections in the United States revolved around issues of potential fraud, with people voting multiple times, stuffing ballot boxes, etc. all of which are well recognized forms of election control. For example, \longversion{Wolchok et al.~}\cite{DBLP:conf/fc/WolchokWIH12} studied security aspects on Internet voting systems.

\begin{table*}[!htbp]
 \begin{center}
 \renewcommand{\arraystretch}{1}
 \begin{tabular}{|c|c|c|c|c|}\hline
  \multirow{2}{*}{Parameters} & \multicolumn{2}{c|}{\OD} & \multicolumn{2}{c|}{\OA}\\\cline{2-5}
  & Scoring rules & Condorcet & Scoring rules & Condorcet\\\hline\hline

  $k_d$ & \WTH [\longversion{\Cref{thm:hard_td}}\shortversion{\Cref{thm:hard_td_od_oa}}] & \WTH [\longversion{\Cref{thm:hard_td_condorcet}}\shortversion{\Cref{thm:hard_td_condorcet_od_oa}}] & \WTH [\longversion{\Cref{thm:hard_td_oa}}\shortversion{\Cref{thm:hard_td_od_oa}}] & \WTH [\longversion{\Cref{cor:hard_ta_condorcet}}\shortversion{\Cref{thm:hard_td_condorcet_od_oa}}] \\\hline

  $k_a$ & \WTH [\longversion{\Cref{thm:hard_ta}}\shortversion{\Cref{thm:hard_ta_sc_con}}] & \WTH [\longversion{\Cref{thm:hard_ta_condorcet}}\shortversion{\Cref{thm:hard_ta_sc_con}}] & \multirow{3}{*}{\WOH [\longversion{\Cref{thm:hard_tad_scoring}}\shortversion{\Cref{thm:hard_tad_scoring_con}}]} & \multirow{3}{*}{\WOH [\longversion{\Cref{thm:hard_tad_condorcet}}\shortversion{\Cref{thm:hard_tad_scoring_con}}]} \\\cline{1-3}

  $(k_a, k_d)$ & \multicolumn{2}{c|}{\makecell{$\OO^*(k_a^{k_d})$ [\Cref{thm:fpt_ka_m}]\\No poly kernel [\Cref{cor:no_poly_kernel}]}} & & \\\hline

  $m$ &\multicolumn{2}{c|}{para-\NPH[\Cref{cor:paranph}]}&\multicolumn{2}{c|}{para-\coNPH[\Cref{cor:paranph}]}\\\hline

 \end{tabular}
 \caption{Summary of parameterized complexity results. $k_d:$ the maximum number of voter groups that the defender can defend. $k_a:$ the maximum number of voter groups that the attacker can attack. $m:$ the number of candidates.}\label{tbl:summary}
 \end{center}
\end{table*}

The study of controlling elections is fundamental to computational social choice: it is widely studied from a theoretical perspective, and has deep practical impact. \shortversion{The pioneering work of }\longversion{Bartholdi et al~}\cite{Bartholdi1992} initiated the study of these problems from a computational perspective, hoping that computational hardness of these problems may suggest a substantial barrier to the phenomena of control: if it is, say \NPH to control an election, then the manipulative agent may not be able to compute an optimal control strategy in a reasonable amount of time. This basic approach has been intensely studied in various other scenarios. For instance, \longversion{Faliszewski et al.~}\cite{FaliszewskiHH11} studied the problem of control where different types of attacks are combined (multimode control), \longversion{Mattei et al~}\cite{MatteiNW14} showed hardness of a variant of control which just exercises different tie-breaking rules, \longversion{Bulteau et al.~}\cite{DBLP:journals/tcs/BulteauCFNT15} studied voter control in a combinatorial setting, etc\longversion{~\cite{DBLP:journals/amai/MiaskoF16,DBLP:journals/tcci/PutF16,DBLP:journals/jair/FaliszewskiHH15,DBLP:conf/aaai/ChenFNT15,DBLP:conf/ecai/MagieraF14,DBLP:journals/jair/FaliszewskiHHR09,DBLP:conf/aaim/FaliszewskiHHR08,DBLP:conf/aaai/FaliszewskiHHR07,DBLP:conf/cats/ErdelyiR10,DBLP:conf/ecai/MaushagenR16,DBLP:journals/mlq/ErdelyiNR09,DBLP:conf/aldt/ErdelyiHH15,DBLP:conf/aldt/ErdelyiHH15,DBLP:conf/ijcai/FitzsimmonsHH13,DBLP:conf/ecai/HemaspaandraHR12a,DBLP:journals/iandc/FaliszewskiHHR11,DBLP:journals/mlq/HemaspaandraHR09,DBLP:conf/tark/FaliszewskiHHR09,DBLP:journals/mst/Menton13,DBLP:conf/ijcai/MentonS13,DBLP:conf/aaai/ParkesX12,DBLP:journals/tcs/Dey18,DBLP:journals/tcs/DeyMN18,DBLP:journals/tcs/DeyMN17,DBLP:conf/mfcs/DeyM17,DBLP:conf/atal/DeyMN15a,DBLP:conf/atal/Dey19}}.

Exploring parameterized complexity of various control problems has also gained a lot of interest. For example, \longversion{Betzler and Uhlmann~}\cite{DBLP:journals/tcs/BetzlerU09} studied parameterized complexity of candidate control in elections and showed interesting connection with digraph problems, \longversion{Liu and Zhu~}\cite{DBLP:journals/ipl/LiuZ10,DBLP:journals/tcs/LiuZ13} studied parameterized complexity of control problem by deleting voters for many common voting rules\longversion{, and so on~\cite{DBLP:journals/tcs/LiuFZL09,DBLP:journals/tcs/WangSYGFSC15,DBLP:conf/atal/HemaspaandraLM13,DBLP:journals/tcs/DeyMN16,DBLP:journals/tcs/DeyMN19}}. Studying election control from a game theoretic approach using security games is also an active area of research. See, for example, the works of \longversion{An et al. and Letchford et al.~}\cite{DBLP:conf/atal/AnBVT13,DBLP:conf/sagt/LetchfordCM09}.


The broad theme of using computational hardness as a barrier to control has two distinct limitations: one is, of course, that some voting rules simply remain computationally vulnerable to many forms of control, in the sense that optimal strategies can be found in polynomial time. The other is that even \NPH control problems often admit reasonable heuristics, can be approximated well, or even admit efficient exact algorithms in realistic scenarios. Therefore, relying on \NP-hardness alone is arguably not a robust strategy against control. To address this issue, the work of \longversion{Yin et al.}~\cite{YinVAH16} explicitly defined the problem of \textit{protecting an election from control}, where in addition to the manipulative agent, we also have a ``defender'', who can also deploy some resources to spoil a planned attack. In this setting, elections are defined with respect to \textit{voter groups} rather than voters, which is a small difference from the traditional control setting. The voter groups model allows us to consider attacks on sets of voters, which is a more accurate model of realistic control scenarios.

In \longversion{Yin et al.~}\cite{YinVAH16}, the defense problem is modeled as a Stackelberg game in which limited protection resources (say $k_d$) are deployed to protect a collection of voter groups and the adversary responds by attempting to subvert the election (by attacking, say, at most $k_a$ groups). They consider the plurality voting rule, and show that the problem of choosing the minimal set of resources that guarantee that an election cannot be controlled is \NPH. They further suggest a Mixed-Integer Program formulation that can usually be efficiently tackled by solvers. Our main contribution is to study this problem in a parameterized setting and provide a refined complexity landscape for it. We also introduce the complementary attack problem, and extend the study to voting rules beyond plurality. We now turn to a summary of our contributions.

\paragraph*{\bf Contribution:} We refer the reader to Section~\ref{sec:prelims} for the relevant formal definitions, while focusing here on a high-level overview of our results. Recall that the \OD problem asks for a set of at most $k_d$ voter groups which, when protected, render any attack on at most $k_a$ voter groups unsuccessful. In this paper, we study the parameterized complexity of \OD{} for all scoring rules and the Condorcet voting rule (these are natural choices because they are computationally vulnerable to control - - the underlying ``attack problem'' can be resolved in polynomial time). We show that the problem of finding an optimal defense is tractable when both the attacker and the defender have limited resources. Specifically, we show that the problem is fixed-parameter tractable with the combined parameter $(k_a,k_d)$ by a natural bounded-depth search tree approach. We also show that the \OD problem is unlikely to admit a polynomial kernel under plausible complexity theoretic assumption. We observe that both these parameters are needed for fixed parameter tractability, as we show \WT-hardness when \OD{} is parameterized by either $k_a$ or $k_d$.

Another popular parameter considered for voting problems is $m$, the number of candidates --- as this is usually small compared to the size of the election in traditional application scenarios. Unfortunately, we show that \OD{} is \NPH even when the election has only $3$ candidates, eliminating the possibility of fixed-parameter algorithms (and even XP algorithms). This strengthens a hardness result shown in\longversion{ Yin et al.}~\cite{YinVAH16}. Our hardness results on a constant number of candidates rely on a succinct encoding of the information about the scores of the candidates from each voter group. We also observe that the problem is polynomially solvable when only two candidates are involved.

We introduce the complementary problem of attacking an election: here the attacker plays her strategy first, and the defender is free to defend any of the attacked groups within the budget. The attacker wins if she is successful in subverting the election no matter which defense is played out. This problem turns out to be harder: it is already \WOH{} when parameterized by \textit{both} $k_a$ and $k_d$, which is in sharp contrast to the \OD{} problem. This problem is also hard in the setting of a constant number of candidates --- specifically, it is \coNPH for the plurality voting rule [\Cref{cor:oa_conph}] and the Condorcet voting rule [\Cref{cor:oa_condorcet_conph}] even when we have only three candidates if every voter group is encoded as the number of plurality votes every candidate receives from that voter group. Our demonstration of the hardness of the attack problem is another step in the program of using computational intractability as a barrier to undesirable phenomenon, which, in this context, is the act of planning a systematic attack on voter groups with limited resources.

We finally propose two simple greedy algorithms for the \OD problem and empirically show that it may be able to solve many instances of practical interest.
%
%

\section{Preliminaries}
\label{sec:prelims}
Let $\CC = \{c_1, c_2, \ldots, c_m\}$ be a set of candidates and $\VV = \{v_1, v_2, \ldots, v_n\}$ a set of voters. If not mentioned otherwise, we denote the set of candidates by \CC, the set of voters by \VV, the number of candidates by $m$, and the number of voters by $n$. Every voter $v_i$ has a preference or vote $\suc_i$ which is a complete order over \CC. We denote the set of all complete orders over \CC by $\LL(\CC)$. We call a tuple of $n$ preferences $(\suc_1, \suc_2, \cdots, \suc_n)\in\LL(\CC)^n$ an $n$-voter preference profile. Often it is convenient to view a preference profile as a multi-set consisting of its votes. The view we are taking will be clear from the context. A voting rule (often called voting correspondence) is a function $r: \cup_{n\in\mathbb{N}} \LL(\CC)^n \longrightarrow 2^\CC\setminus\{\emptyset\}$ which selects, from a preference profile, a nonempty set of candidates as the winners. We refer the reader to~\cite{brandt2015handbook} for a comprehensive introduction to computational social choice. In this paper we will be focusing on two voting rules -- the scoring rules and the Condorcet voting rule which are defined as follows.

{\bf Scoring Rule:} A collection of $m$-dimensional vectors $\overrightarrow{s_m}=\left(\alpha_1,\alpha_2,\dots,\alpha_m\right)\in\mathbb{R}^m$  with $\alpha_1\ge\alpha_2\ge\dots\ge\alpha_m$ and $\alpha_1>\alpha_m$ for every $m\in \mathbb{N}$ naturally defines a voting rule --- a candidate gets score $\alpha_i$ from a vote if it is placed at the $i^{th}$ position, and the  score of a candidate is the sum of the scores it receives from all the votes. The winners are the candidates with the highest score. Given a set of candidates \CC, a score vector $\overrightarrow{\alpha}$ of length $|\CC|$, a candidate $x\in\CC$, and a profile \PP, we denote the score of $x$ in \PP by $s^{\overrightarrow{\alpha}}_\PP(x)$. When the score vector $\overrightarrow{\alpha}$ is clear from the context, we omit $\overrightarrow{\alpha}$ from the superscript. A straight forward observation is that the scoring rules remain unchanged if we multiply every $\alpha_i$ by any constant $\lambda>0$ and/or add any constant $\mu$. Hence, we assume without loss of generality that for any score vector $\overrightarrow{s_m}$, there exists a $j$ such that $\alpha_j - \alpha_{j+1}=1$ and $\alpha_k = 0$ for all $k>j$. We call such a score vector a {\em normalized score vector}.

{\bf Weighted Majority Graph and Condorcet Voting Rule:} Given an election $\EE = (\CC, (\suc_1, \suc_2, \ldots, \suc_n))$ and two candidates $x, y\in\CC$, let us define $N_\EE(x,y)$ to be the number of votes where the candidate $x$ is preferred over $y$. We say that a candidate $x$ defeats another candidate $y$ in {\em pairwise election} if $N_\EE(x,y) > N_\EE(y,x)$. Using the election \EE, we can construct a weighted directed graph $\GG_\EE = (\UU = \CC, E)$ as follows. The vertex set \UU of the graph $\GG_\EE$ is the set of candidates $\CC$. For any two candidates $x, y\in\CC$ with $x\ne y$, let us define the margin $\DD_\EE(x, y)$ of $x$ from $y$ to be $N_\EE(x,y) - N_\EE(y,x)$. We have an edge from $x$ to $y$ in $\GG_\EE$ if $\DD_\EE(x,y)>0$. Moreover, in that case, the weight $w(x,y)$ of the edge from $x$ to $y$ is $\DD_\EE(x,y)$. A candidate $c$ is called the {\em Condorcet winner} of an election \EE if there is an edge from $c$ to every other vertices in the weighted majority graph $\GG_\EE$. The Condorcet voting rule outputs the Condorcet winner if it exists and outputs the set \CC of all candidates otherwise.

Let $r$ be a voting rule. We study the $r$-\OD problem which was defined by \longversion{Yin et al.~}\cite{YinVAH16}. It is defined as follows. Intuitively, the $r$-\OD problem asks if there is a way to defend $k_d$ voter groups such that, irrespective of which $k_a$ voter groups the attacker attacks, the output of the election (that is the winning set of candidates) is always same as the original one. A voter group gets deleted if only if it is attacked but not defended.

\begin{definition}[$r$-\OD]
 Given $n$ voter groups $\GG_i, i\in[n],$ two integers $k_a$ and $k_d$, does there exist an index set $\II\subseteq [n]$ with $|\II|\le k_d$ such that, for every $\II^\pr\subset [n]\setminus\II$ with $|\II^\pr|\le k_a$, we have $r((\GG_i)_{i\in[n]\setminus\II^\pr}) = r((\GG_i)_{i\in[n]})$? The integers $k_a$ and $k_d$ are called respectively attacker's resource and defender's resource. We denote an arbitrary instance of the $r$-\OD problem by $(\CC, \{\GG_i: i\in[n]\}, k_a, k_d)$.
\end{definition}

We also study the $r$-\OA problem which is defined as follows. Intuitively, in the $r$-\OA problem the attacker is interested to know if it is possible to attack $k_a$ voter groups such that, no matter which $k_d$ voter groups the defender defends, the outcome of the election is never same as the original (that is the attack is successful).

\begin{definition}[$r$-\OA]
 Given $n$ voter groups $\GG_i, i\in[n],$ two integers $k_a$ and $k_d$, does there exist an index set $\II\subseteq [n]$ with $|\II|\le k_a$ such that, for every $\II^\pr\subseteq [n]$ with $|\II^\pr|\le k_d$, we have $r((\GG_i)_{i\in [n]\setminus (\II\setminus\II^\pr)}) \ne r((\GG_i)_{i\in[n]})$? We denote an arbitrary instance of the $r$-\OA problem by $(\CC, \{\GG_i: i\in[n]\}, k_a, k_d)$.
\end{definition}


{\bf Encoding of the Input Instance:} In both the $r$-\OD and $r$-\OA problems, we assume that every input voter group \GG is encoded as follows. The encoding lists all the different votes \suc that appear in the voter group \GG along with the number of times the vote \suc appear in \GG. Hence, if a voter group \GG contains only $k$ different votes over $m$ candidates and consists of $n$ voters, then the encoding of \GG takes $\OO(km\log m \log n)$ bits of memory.

{\bf Parameterized complexity:}
\longversion{
In parameterized
complexity, each problem instance comes
with a parameter $k$. Formally, a parameterized problem $\Pi$ is a
subset of $\Gamma^{*}\times
\mathbb{N}$, where $\Gamma$ is a finite alphabet. An instance of a
parameterized problem is a tuple $(x,k)$, where $k$ is the
parameter. A central notion is \emph{fixed parameter
tractability} (FPT) which means, for a
given instance $(x,k)$, solvability in time $f(k) \cdot p(|x|)$,
where $f$ is an arbitrary function of $k$ and
$p$ is a polynomial in the input size $|x|$.
Just as NP-hardness is used as evidence that a problem probably is not polynomial time solvable, there exists a hierarchy of complexity classes above FPT, and showing that a parameterized problem is hard for one of these classes is considered
evidence that the problem is unlikely to be fixed-parameter tractable. The main classes in this hierarchy are: $ \FPT  \subseteq \WONE \subseteq \WTWO \subseteq \cdots \subseteq \WP \subseteq \XP.$ We now define the notion of parameterized reduction~\cite{CyganEtAl}.}
\shortversion{A parameterized problem $\Pi$ is a
subset of $\Gamma^{*}\times
\mathbb{N}$, where $\Gamma$ is a finite alphabet. A central notion is \emph{fixed parameter
tractability} (FPT) which means, for a
given instance $(x,k)$, solvability in time $f(k) \cdot p(|x|)$,
where $f$ is an arbitrary function of $k$ and
$p$ is a polynomial in the input size $|x|$. There exists a hierarchy of complexity classes above FPT, and showing that a parameterized problem is hard for one of these classes is considered
evidence that the problem is unlikely to be fixed-parameter tractable. The main classes in this hierarchy are: $ \FPT  \subseteq \WONE \subseteq \WTWO \subseteq \cdots \subseteq \WP \subseteq \XP.$ We now define the notion of parameterized reduction~\cite{CyganEtAl}.}
\begin{definition}
Let $A,B$ be parameterized problems.  We say that $A$ is {\bf \em fpt-reducible} to $B$ if there exist functions
$f,g:\mathbb{N}\rightarrow \mathbb{N}$, a constant $\alpha \in \mathbb{N}$ and
 an algorithm $\Phi$ which transforms an instance $(x,k)$ of $A$ into an instance $(x',g(k))$ of $B$
in time $f(k) |x|^{\alpha}$
so that $(x,k) \in A$ if and only if $(x',g(k)) \in B$.
\end{definition}

To show W-hardness\longversion{ in the parameterized setting}, it is enough to give a parameterized reduction from a known hard problem.\longversion{ For a more detailed and formal introduction to parameterized complexity, we refer the reader to~\cite{CyganEtAl} for a detailed introduction to this paradigm.}

\longversion{

\begin{definition}{\rm \bf[Kernelization]}~\cite{niedermeier2006invitation,flum2006parameterized}
A kernelization algorithm for a parameterized problem   $\Pi\subseteq \Gamma^{*}\times \mathbb{N}$ is an
algorithm that, given $(x,k)\in \Gamma^{*}\times \mathbb{N} $, outputs, in time polynomial in $|x|+k$, a pair
$(x',k')\in \Gamma^{*}\times
  \mathbb{N}$ such that (a) $(x,k)\in \Pi$ if and only if
  $(x',k')\in \Pi$ and (b) $|x'|,k'\leq g(k)$, where $g$ is some
  computable function.  The output instance $x'$ is called the
  kernel, and the function $g$ is referred to as the size of the
  kernel. If $g(k)=k^{O(1)}$, then we say that
  $\Pi$ admits a polynomial kernel.
\end{definition}
For many parameterized problems, it is well established that the existence of a polynomial kernel would
imply the collapse of the polynomial hierarchy to the third level (or more precisely, \caveat{}).
Therefore, it is considered unlikely that these problems would admit polynomial-sized kernels.
For showing kernel lower bounds, we simply establish reductions from these problems.
}

\longversion{
\begin{definition}{\rm \bf[Polynomial Parameter Transformation]}
\label{def:ppt-reduction} {\rm\cite{BodlaenderThomasseYeo2009}}
Let $\Gamma_1$ and $\Gamma_2$ be parameterized problems. We say that $\Gamma_1$ is
polynomial time and parameter reducible to $\Gamma_2$, written
$\Gamma_1\le_{Ptp} \Gamma_2$, if there exists a polynomial time computable
function $f:\Sigma^{*}\times\mathbb{N}\to\Sigma^{*}\times\mathbb{N}$, and a
polynomial $p:\mathbb{N}\to\mathbb{N}$, and for all
$x\in\Sigma^{*}$ and $k\in\mathbb{N}$, if
$f\left(\left(x,k\right)\right)=\left(x',k'\right)$, then
$\left(x,k\right)\in \Gamma_1$ if and only if $\left(x',k'\right)\in \Gamma_2$,
and $k'\le p\left(k\right)$. We call $f$ a polynomial parameter
transformation (or a PPT) from $\Gamma_1$ to $\Gamma_2$.
\end{definition}

This notion of a reduction is useful in showing kernel lower bounds
because of the following theorem.

\begin{theorem}\label{thm:ppt-reduction}~{\rm\cite[Theorem 3]{BodlaenderThomasseYeo2009}}
  Let $P$ and $Q$ be parameterized problems whose derived
  classical problems are $P^{c},Q^{c}$, respectively. Let $P^{c}$
  be $\NPC$, and $Q^{c}\in$ \NP. Suppose there exists a PPT from
  $P$ to $Q$.  Then, if $Q$ has a polynomial kernel, then $P$ also
  has a polynomial kernel.
\end{theorem}
}

\section{Classical Complexity Results}

\longversion{Yin et al.~}\cite{YinVAH16} showed that the \OD problem is polynomial time solvable for the plurality voting rule when we have only $2$ candidates. On the other hand, they also showed that the \OD problem is \NPC when we have an {\em unbounded} number of candidates. We begin with improving their \NP-completeness result by showing that the \OD problem becomes \NPC even when we have only $3$ candidates and the attacker can attack any number of voter groups. Towards that, we reduce the $k$-{\sc Sum} problem to the \OD problem. The $k$-{\sc Sum} problem is defined as follows.

\begin{definition}[$k$-{\sc Sum}]
 Given a set of $n$ positive integers $\WW = \{w_i, i\in[n]\},$ and two positive integers $k\le n$ and $M$, does there exist an index set $\II\subset[n]$ with $|\II|= k$ such that $\sum_{i\in\II} w_i = M$?
\end{definition}

The $k$-{\sc Sum} problem can be easily proved to be \NPC by modifying the \NP-completeness proof of the {Subset Sum} problem in \longversion{Cormen et al.~}\cite{Cormen}.
We also need the following structural result for normalized scoring rules which has been used before~\cite{baumeister2011computational,journalsDeyMN16}.

\begin{lemma}\label{lem:scoringrule}
 Let $\mathcal{C} = \{c_1, \ldots, c_m\}$ be a set of candidates and $\overrightarrow{\alpha}$ a normalized score vector of length $|\mathcal{C}|$. Let $x, y\in\CC, x\ne y,$ be any two arbitrary candidates. Then there exists a profile $\PP_x^y$ consisting of $m$ votes such that we have the following.\\
 $s_{\PP_x^y}(x)+1 = s_{\PP_x^y}(y)-1 = s_{\PP_x^y}(a) \text{ for every } a\in\CC\setminus\{x,y\}$
\end{lemma}

For any two candidates $x,y\in\CC, x\ne y$, we use $\PP_x^y$ to denote the profile as defined in \Cref{lem:scoringrule}. We are now ready to present our \NP-completeness result for the \OD problem for the scoring rules even in the presence of $3$ candidates only. In the interest of space, we will provide only a sketch of a proof for a several results.

\begin{theorem}\label{thm:npc}
 The \OD problem is \NPC for every scoring rule even if the number of candidates is $3$ and the attacker can attack any number of the voter groups.
\end{theorem}

\begin{proof}
 The \OD problem for every scoring rule can be shown to belong to \NP by using a defense strategy~$S$ (a subset of at most $k_d$ voter groups) as a certificate. The fact that the certificate can be validated in polynomial time involves checking if there exists a successful attack despite protecting all groups in $S$. This can be done in polynomial time, but due to space constraints, we defer a detailed argument to a full version of this manuscript. We now turn to the reduction from \kSum{}.

  Let $\overrightarrow{\alpha}$ be any normalized score vector of length $3$. The \OD problem for the scoring rule based on $\overrightarrow{\alpha}$ belongs to \NP. Let $(\WW=\{w_1, \ldots, w_n\}, k, M)$ be an arbitrary instance of the \kSum problem. We can assume, without loss of generality, that $8$ divides $M$ and $w_i$ for every $i\in[n]$; if this is not the case, we replace $M$ and $w_i$ by respectively $8M$ and $8w_i$ for every $i\in[n]$ which clearly is an equivalent instance of the original instance. Let us also assume, without loss of generality, that $2k<n$ (if not then add enough copies of $M+1$ to \WW) and $M<\sum_{i=1}^n w_i$ (since otherwise, it is a trivial \NO instance). We construct the following instance of the \OD problem for the scoring rule based on $\overrightarrow{\alpha}$. Let $M^\pr$ be an integer such that $M^\pr > \sum_{i=1}^n w_i$ and $8$ divides $M^\pr$. We have $3$ candidates, namely $a$, $b$, and $c$. We have the following voter groups.

 \begin{itemize}[topsep=0pt,itemsep=0pt,leftmargin=*]
  \item For every $i\in[n]$, we have a voter group $\GG_i$ consisting of $w_i$ copies of $\PP_a^c$ (as defined in \Cref{lem:scoringrule}) and $M^\pr-w_i$ copies of $\PP_b^c$. Hence, we have the following.\\
  $s_{\GG_i}(c) = s_{\GG_i}(a) + M^\pr + w_i = s_{\GG_i}(b) + 2M^\pr - w_i$
  \item We have one voter group $\hat{\GG}$ consisting of $\nfrac{(kM^\pr+M)}{2}-3$ copies of $\PP_c^a$, $\nfrac{(kM^\pr-M)}{2}-1$ copies of $\PP_c^b$, and $\nfrac{(kM^\pr-M)}{2}-1$ copies of $\PP_a^b$. We have the following.\\
  $s_{\hat{\GG}}(c) = s_{\hat{\GG}}(a) - (kM^\pr + M -6) = s_{\hat{\GG}}(b) - (2kM^\pr-M-6)$
 \end{itemize}

 Let \QQ be the resulting profile; that is $\QQ = \cup_{i=1}^n \GG_i \cup \hat{\GG}$. We have $s_{\QQ}(c) = s_{\QQ}(a) + (n-k)M^\pr +\sum_{i=1}^n w_i - M + 6 = s_{\QQ}(b) + (n-2k)M^\pr + M - \sum_{i=1}^n w_i +6$. Since $n>2k$ and $M^\pr > \sum_{i=1}^n w_i$, we have $s_{\QQ}(c) > s_{\QQ}(a)$ and $s_{\QQ}(c)>s_{\QQ}(b)$. Thus the candidate $c$ wins the election uniquely. We define $k_d$, the maximum number of voter groups that the defender can defend, to be $k$. We define $k_a$, the maximum number of voter groups that the attacker can attack, to be $n+1$. This finishes the description of the \OD instance. We claim that the two instances are equivalent.

 In the forward direction, let the \kSum instance be a \YES instance and $\II\subset [n]$ with $|\II|= k$ be an index set such that $\sum_{i\in\II} w_i = M$. Let us consider the defense strategy where the defender protects the voter groups $\GG_i$ for every $i\in\II$. Since $\sum_{i\in\II} w_i = M$, we have $\sum_{i\in\II} (M^\pr-w_i) =  kM^\pr -M$. Let \HH be the profile of voter groups corresponding to the index set \II; that is, $\HH = \cup_{i\in\II} \GG_i$. Let $\HH^\pr$ be the profile remaining after the attacker attacks some voter groups. Without loss of generality, we can assume that the attacker does not attack the voter group $\hat{\GG}$ since otherwise the candidate $c$ continues to win uniquely. We thus obviously have $\HH\cup\hat{\GG}\subseteq\HH^\pr$. We have $s_{\HH\cup\hat{\GG}}(c) = s_{\HH\cup\hat{\GG}}(a) + kM^\pr + \sum_{i\in\II} w_i - (kM^\pr + M -6) = s_{\HH\cup\hat{\GG}}(a) + 6$ and $s_{\HH\cup\hat{\GG}}(c) = s_{\HH\cup\hat{\GG}}(b) + 2kM^\pr - \sum_{i\in\II}w_i - (2kM^\pr-M-6) = s_{\HH\cup\hat{\GG}}(b) + 6$. Since the candidate $c$ receives as much score as any other candidate in the voter group $\GG_i$ for every $i\in[n]$, we have $s_{\HH^\pr\cup\hat{\GG}}(c) \ge s_{\HH^\pr\cup\hat{\GG}}(a) + 6$ and $s_{\HH^\pr\cup\hat{\GG}}(c) \ge s_{\HH^\pr\cup\hat{\GG}}(b) + 6$. Hence, the candidate $c$ wins uniquely in the resulting profile $\HH^\pr$ after the attack and thus the defense is successful.

 In the other direction, let the \OD instance be a \YES instance. Without loss of generality, we can assume that the attacker does not attack the voter group $\hat{\GG}$ and thus the defender does not defend the voter group $\hat{\GG}$. We can also assume, without loss of generality, that the defender defends exactly $k$ voter groups since the candidate $c$ receives as much score as any other candidate in the voter group $\GG_i$ for every $i\in[n]$. Let $\II\subset[n]$ with $|\II|=k$ such that defending all the voter groups $\GG_i, i\in\II$ is a successful defense strategy. We claim that $\sum_{i\in\II} w_i \ge M$. Suppose not, then let us assume that $\sum_{i\in\II} w_i < M$. Since, $w_i$ is divisible by $8$ and positive for every $i\in[n]$ and $m$ is divisible by $8$, we have $\sum_{i\in\II} w_i \le M-8$. Let \HH be the profile of voter groups corresponding to the index set \II; that is, $\HH = \cup_{i\in\II} \GG_i$. We have $s_{\HH\cup\hat{\GG}}(c) = s_{\HH\cup\hat{\GG}}(a) + kM^\pr + \sum_{i\in\II} w_i - (kM^\pr + M -6) \le s_{\HH\cup\hat{\GG}}(a) + M - 8 - M + 6 = s_{\HH\cup\hat{\GG}}(a) - 2$. Hence attacking the voter groups $\GG_i, i\in[n]\setminus\II$ makes the score of $c$ strictly less than the score of $a$. This contradicts our assumption that defending all the voter groups $\GG_i, i\in\II$ is a successful defense strategy. Hence we have $\sum_{i\in\II} w_i \ge M$. We now claim that $\sum_{i\in\II} w_i \le M$. Suppose not, then let us assume that $\sum_{i\in\II} w_i > M$. Since, $w_i$ is divisible by $8$ and positive for every $i\in[n]$ and $m$ is divisible by $8$, we have $\sum_{i\in\II} w_i \ge M+8$. Let $\HH^\pr$ be the profile of voter groups corresponding to the index set \II; that is, $\HH^\pr = \cup_{i\in\II} \GG_i$. We have $s_{\HH^\pr\cup\hat{\GG}}(c) = s_{\HH^\pr\cup\hat{\GG}}(b) + 2kM^\pr - \sum_{i\in\II}w_i - (2kM^\pr-M-6) \le s_{\HH^\pr\cup\hat{\GG}}(b) - (M+8) + M + 6 = s_{\HH^\pr\cup\hat{\GG}}(b) - 2$. Hence attacking the voter groups $\GG_i, i\in[n]\setminus\II$ makes the score of $c$ strictly less than the score of $b$. This contradicts our assumption that defending all the voter groups $\GG_i, i\in\II$ is a successful defense strategy. Hence we have $\sum_{i\in\II} w_i \le M$. Therefore we have $\sum_{i\in\II} w_i = M$ and thus the \kSum instance is a \YES instance.
\end{proof}

In the proof of \Cref{thm:npc}, we observe that the reduced instance of the \OD problem viewed as an instance of the \OA problem is a \NO instance if and only if the \kSum instance is a \YES instance. Hence, the same reduction as in the proof of \Cref{thm:npc} gives us the following result for the \OA problem.

\begin{corollary}\label{cor:oa_conph}
 The \OA problem is \coNPH for every scoring rule even if the number of candidates is $3$ and the attacker can attack any number of voter groups.
\end{corollary}

We now prove a similar hardness result as of \Cref{thm:npc} for the Condorcet voting rule.

\begin{theorem}\label{thm:oa_condorcet_npc}
 The \OD problem is \NPC for the Condorcet voting rule even if the number of candidates is $3$ and the attacker can attack any number of voter groups.
\end{theorem}
\longversion{
\longversion{\begin{proof}}
\shortversion{\begin{proof_sketch}}
 The \OD problem for the Condorcet voting rule clearly belongs to \NP. To show \NP-hardness, we reduce an arbitrary instance of the \kSum problem to the \OD problem for the Condorcet voting rule. Let $(\{w_1, \ldots, w_n\}, k, M)$ be an arbitrary instance of the \kSum problem. We construct the following instance of the \OD problem for the Condorcet voting rule. Let $M^\pr = \max\{w_i:i\in[n]\}$. We have $3$ candidates, namely $a$, $b$, and $c$. We have the following voter groups.
 \begin{itemize}[topsep=0pt,itemsep=0pt]
  \item For every $i\in[n]$, we have a voter group $\GG_i$ where $\DD_{\GG_i} (a,b) = 2w_i, \DD_{\GG_i}(a,c)=2(M^\pr-w_i)$, and $\DD_{\GG_i}(b,c) = 0$.
  \item We have one voter group $\hat{\GG}$ where the candidates $b$ and $c$ receive respectively $\DD_{\hat{\GG}} (b,a) = 2M-1, \DD_{\hat{\GG}} (c,a) = 2(kM^\pr-M)-1$, and $\DD_{\hat{\GG}}(b,c)=1$.
 \end{itemize}
 We define $k_d$, the maximum number of voter groups that the defender can defend, to be $k$. We define $k_a$, the maximum number of voter groups that the attacker can attack, to be $n+1$. We observe that the candidate $a$ is the Condorcet winner of the election. This finishes the description of the \OD instance. \shortversion{The proof of equivalence of the two instances is similar in spirit as the proof of \Cref{thm:npc}.}\longversion{We claim that the two instances are equivalent.

 In the forward direction, let the \kSum instance be a \YES instance and $\II\subset [n]$ with $|\II|= k$ be an index set such that $\sum_{i\in\II} w_i = M$. Let us consider the defense strategy where the defender protects the voter groups $\GG_i$ for every $i\in\II$. Since $\sum_{i\in\II} w_i = M$, we have $\sum_{i\in\II} (M^\pr-w_i) =  kM^\pr -M$. Without loss of generality, we can assume that the attacker does not attack the voter group $\hat{\GG}$. We observe that the candidate $a$ is the Condorcet winner of the election even when the attacker attacks all the voter groups $\GG_j, j\in[n]\setminus\II$. Hence the \OD instance is a \YES instance.

 In the other direction, let the \OD instance be a \YES instance. Without loss of generality, we can assume that the attacker does not attack the voter group $\hat{\GG}$ and thus the defender does not defend the voter group $\hat{\GG}$. We can also assume, without loss of generality, that the defender defends exactly $k$ voter groups since the candidate $a$ continues to be the Condorcet winner if the attacker attacks at most $k-1$ voter groups. Let $\II\subset[n]$ with $|\II|=k$ such that defending all the voter groups $\GG_i, i\in\II$ is a successful defense strategy. We claim that $\sum_{i\in\II} w_i \ge M$. Suppose not, then let us assume that $\sum_{i\in\II} w_i < M$. Then attacking the voter groups $\GG_i, i\in[n]\setminus\II$ makes the candidate $b$ defeat the candidate $a$ in pairwise election. This contradicts or assumption that defending all the voter groups $\GG_i, i\in\II$ is a successful defense strategy. Hence we have $\sum_{i\in\II} w_i \ge M$. We now claim that $\sum_{i\in\II} w_i \le M$. Suppose not, then let us assume that $\sum_{i\in\II} w_i > M$. Then attacking the voter groups $\GG_i, i\in[n]\setminus\II$ makes the candidate $c$ defeat the candidate $a$ in pairwise election. This contradicts or assumption that defending all the voter groups $\GG_i, i\in\II$ is a successful defense strategy. Hence we have $\sum_{i\in\II} w_i \le M$. Therefore we have $\sum_{i\in\II} w_i = M$ and thus the \kSum instance is a \YES instance.}
\longversion{\end{proof}}
\shortversion{\end{proof_sketch}}
}
In the proof of \Cref{thm:oa_condorcet_npc}, we observe that the reduced instance of \OD viewed as an instance of the \OA problem is a \NO instance if and only if the \kSum instance is a \YES instance. Hence, the same reduction as in the proof of \Cref{thm:oa_condorcet_npc} gives us the following result for the \OA problem.

\begin{corollary}\label{cor:oa_condorcet_conph}
 The \OA problem is \coNPH for the Condorcet voting rule even if the number of candidates is $3$ and the attacker can attack any number of voter groups.
\end{corollary}

\section{W-Hardness Results}

In this section, we present our hardness results for the \OD and the \OA problems in the parameterized complexity framework. We consider the following parameters for both the problems -- number of candidate ($m$), defender's resource ($k_d$), and attacker's resource ($k_a$). From \Cref{thm:npc,cor:oa_conph,thm:oa_condorcet_npc,cor:oa_condorcet_conph} we immediately have the following result for the \OD and \OA problems parameterized by the number of candidates for both the scoring rules and the Condorcet voting rule.

\begin{corollary}\label{cor:paranph}
 The \OD problem is para-\NPH parameterized by the number of candidates for both the scoring rules and the Condorcet voting rule. The \OA problem is para-\coNPH parameterized by the number of candidates for both the scoring rules and the Condorcet voting rule.
\end{corollary}

The \NP-completeness proof for the \OD problem for the plurality voting rule by\longversion{ Yin et al.}~\cite{YinVAH16} is actually a parameter preserving reduction from the \HS problem parameterized by the solution size.\longversion{ The \HS problem is defined as follows.

\begin{definition}[\HS]
 Given a universe \UU, a set $\SS = \{S_i: i\in[t]\}$ of subsets of \UU, and a positive integer $k$ which is at most $|U|$, does there exist a subset $\WW\subseteq\UU$ with $|\WW|= k$ such that $\WW\cap S_i \ne \emptyset$ for every $i\in[t]$. We denote an arbitrary instance of \HS by $(\UU, \SS, k)$.
\end{definition}}
Since the \HS problem parameterized by the solution size $k$ is known to be \WTC~\cite{downey1999parameterized}, the following result immediately follows from Theorem 2 of\longversion{ Yin et al.}~\cite{YinVAH16}.

\begin{observation}[\cite{YinVAH16}] The \OD problem for the plurality voting rule is \WTH parameterized by $k_d$.\label{obs:hard_td}
\end{observation}

We now generalize \Cref{obs:hard_td} to any scoring rule by exhibiting a polynomial parameter transform from the \HS problem parameterized by the solution size.

\shortversion{
\begin{theorem}\shortversion{[$\star$]}
The \OD and \OA problems for every scoring rule is \WTH parameterized by $k_d$.\label{thm:hard_td_od_oa}
\end{theorem}}

\longversion{

\begin{theorem}\shortversion{[$\star$]}
The \OD problem for every scoring rule is \WTH parameterized by $k_d$.\label{thm:hard_td}
\end{theorem}
\longversion{
\begin{proof}
 Let $(\UU, \SS = \{S_j: j\in[t]\}, k)$ be an arbitrary instance of \HS. Let $\UU=\{z_i:i\in[n]\}$. Without loss of generality, we assume that $S_j \ne \emptyset$ for every $j\in[t]$ since otherwise the instance is a \NO instance. Let $\overrightarrow{\alpha}$ be a normalized score vector of length $t+2$. We construct the following instance of the \OD problem for the scoring rule based on $\overrightarrow{\alpha}$. The set of candidates $\CC = \{x_j: j\in[t]\} \cup \{y,d\}$. We have the following voter groups.

 \begin{itemize}[topsep=0pt,itemsep=0pt]
  \item For every $i\in[n]$, we have a voter group $\GG_i$. For every $j\in[t]$ with $z_i\in S_j$ we have $2$ copies of $\PP_{x_j}^d$ in $\GG_i$.
  \item We have one group $\hat{\GG}$ where we have $2tn$ copies of $\PP_d^{x_j}$ for every $j\in[n]$ and $2tn-1$ copies of $\PP_d^y$.
 \end{itemize}

  Let \QQ be the resulting profile; that is $\QQ=\cup_{i=1}^n \GG_i\cup \hat{\GG}$. We define the defender's resource $k_d$ to be $k+1$ and attacker's resource to be $n$. This finishes the description of the \OD instance. Since $S_j\ne \emptyset$ for every $j\in[t]$, we have $s_\QQ(y)>s_\QQ(x_j)$ for every $j\in[t]$. We also have $s_\QQ(y)>s_\QQ(d)$. Hence the candidate $y$ is the unique winner of the profile \QQ. We now prove that the \OD instance $(\CC, \QQ, k_a, k_d)$ is equivalent to the \HS instance $(\UU, \SS, k)$.

  In the forward direction, let us suppose that the \HS instance is a \YES instance. Let $\II\subset[n]$ be such that $|\II|=k$ and $\{z_i:i\in\II\} \cap S_j \ne \emptyset$. We claim that the defender's strategy of defending the voter groups $\GG_j$ for every $j\in[t]\setminus\II$ and $\hat{\GG}$ results in a successful defense. Let \HH be the profile of voter groups corresponding to the index set \II; that is, $\HH = \cup_{i\in\II} \GG_i$. Let $\HH^\pr$ be the profile remaining after the attacker attacks some voter groups. We thus obviously have $\HH\cup\hat{\GG}\subseteq\HH^\pr$. Since $\{z_i:i\in\II\}$ forms a hitting set, we have $s_{\HH^\pr}(y)>s_{\HH^\pr}(x_j)$ for every $j\in[t]$. Also since the voter group $\hat{\GG}$ is defended, we have $s_{\HH^\pr}(y)>s_{\HH^\pr}(d)$. Hence the candidate $y$ continues to win uniquely even after the attack and hence the \OD instance is a \YES instance.

  In the other direction, let the \OD instance be a \YES instance. Without loss of generality, we can assume that the defender defends the voter group $\hat{\GG}$ since otherwise the attacker can attack the voter group $\hat{\GG}$ which makes the score of the candidate $d$ more than the score of the candidate $y$ and thus defense would fail. We can also assume, without loss of generality, that the defender defends exactly $k$ voter groups. Let $\II\subset[n]$ with $|\II|=k$ such that defending all the voter groups $\GG_i, i\in\II$ and $\hat{\GG}$ is a successful defense strategy. Let us consider $\ZZ=\{z_i:i\in\II\}\subseteq\UU$. We claim that $\ZZ$ must form a hitting set. Indeed, otherwise let us assume that there exists a $j\in[t]$ such that $\ZZ\cap S_j = \emptyset$. Consider the situation where the attacker attacks voter groups $\GG_i$ for every $i\in[n]\setminus\II$. We observe that $s_{\cup_{i\in\II}\GG_i \cup \hat{\GG}}(x_j) > s_{\cup_{i\in\II}\GG_i \cup \hat{\GG}}(y)$. This contradicts our assumption that defending all the voter groups $\GG_i, i\in\II$ and $\hat{\GG}$ is a successful defense strategy. Hence \ZZ forms a hitting set and thus the \HS instance is a \YES instance.
\end{proof}
}

In the proof of \Cref{thm:hard_td}, we observe that the reduced instance of \OD viewed as an instance of the \OA problem is a \NO instance if and only if the \kSum instance is a \YES instance. Hence, the same reduction as in the proof of \Cref{thm:hard_td} gives us the following result for the \OA problem.

\begin{corollary}
 The \OA problem for every scoring rule is \WTH parameterized by $k_d$.\label{thm:hard_td_oa}
\end{corollary}

\longversion{
We now show \WT-hardness of the \OD problem for the Condorcet voting rule parameterized by $k_d$. Towards that, we need the following lemma which has been used before
\cite{mcgarvey1953theorem,xia2008determining}.
\begin{lemma}\label{thm:mcgarvey}
 For any function $f:\mathcal{C} \times \mathcal{C} \longrightarrow \mathbb{Z}$, such that
 \begin{enumerate}[topsep=0pt,itemsep=0pt]
  \item $\forall a,b \in \mathcal{C}, f(a,b) = -f(b,a)$.
  \item $\forall a,b, c, d \in \mathcal{C}, f(a,b) + f(c,d)$ is even,
 \end{enumerate}
 there exists a $n$ voters' profile such that for all $a,b \in \mathcal{C}$, $a$ defeats
 $b$ with a margin of $f(a,b)$. Moreover,
 $$n \text{ is even and } n = O\left(\sum_{\{a,b\}\in \mathcal{C}\times\mathcal{C}} |f(a,b)|\right)$$
\end{lemma}
}
}

\shortversion{
Next, we show the \WT-hardness of the \OD and \OA problems for the Condorcet voting rule parameterized by $k_d$. This is also a parameter-preserving reduction from the \HS{} problem.

\begin{theorem}\shortversion{[$\star$]}
 The \OD and \OA problems for the Condorcet voting rule is \WTH parameterized by $k_d$.\label{thm:hard_td_condorcet_od_oa}
\end{theorem}
}

\longversion{
Next, we show the \WT-hardness of the \OD problem for the Condorcet voting rule parameterized by $k_d$. This is also a parameter-preserving reduction from the \HS{} problem.

\begin{theorem}\shortversion{[$\star$]}
 The \OD problem for the Condorcet voting rule is \WTH parameterized by $k_d$.\label{thm:hard_td_condorcet}
\end{theorem}

\longversion{\begin{proof}
 Let $(\UU, \SS = \{S_j: j\in[t]\}, k)$ be an arbitrary instance of \HS. Let $\UU=\{z_i:i\in[n]\}$. Without loss of generality, we assume that $S_j \ne \emptyset$ for every $j\in[t]$ since otherwise the instance is a \NO instance. We construct the following instance of the \OD problem for the Condorcet voting rule. The set of candidates $\CC = \{x_j: j\in[t]\} \cup \{y\}$. For every $i\in[n]$, we have a voter group $\GG_i$. For every $j\in[t]$ with $z_i\in S_j$ we have $\DD_{\GG_i}(y,x_j)=2$. Let \QQ be the resulting profile; that is $\QQ=\cup_{i=1}^n \GG_i$. We define the defender's resource $k_d$ to be $k$ and attacker's resource to be $n$. This finishes the description of the \OD instance. Since $S_j\ne \emptyset$ for every $j\in[t]$, we have $\DD_{\QQ}(y,x_j)\ge 2$ for every $j\in[t]$. Hence the candidate $y$ is the Condorcet winner of the profile \QQ. \shortversion{The proof of equivalence of the two instances is similar in spirit as the proof of \Cref{thm:hard_td}.}\longversion{We now prove that the \OD instance $(\CC, \QQ, k_a, k_d)$ is equivalent to the \HS instance $(\UU, \SS, k)$.

 In the forward direction, let us suppose that the \HS instance is a \YES instance. Let $\II\subset[n]$ be such that $|\II|=k$ and $\{z_i:i\in\II\} \cap S_j \ne \emptyset$. We claim that the defender's strategy of defending the voter groups $\GG_j$ for every $j\in[t]\setminus\II$ results in a successful defense. Let \HH be the profile of voter groups corresponding to the index set \II; that is, $\HH = \cup_{i\in\II} \GG_i$. Let $\HH^\pr$ be the profile remaining after the attacker attacks some voter groups. We thus obviously have $\HH\subseteq\HH^\pr$. Since $\{z_i:i\in\II\}$ forms a hitting set, we have $\DD_{\HH^\pr}(y,x_j)\ge 2$ for every $j\in[t]$. Hence the candidate $y$ continues to win uniquely even after the attack and hence the \OD instance is a \YES instance.

 In the other direction, let the \OD instance be a \YES instance. We can also assume, without loss of generality, that the defender defends exactly $k$ voter groups. Let $\II\subset[n]$ with $|\II|=k$ such that defending all the voter groups $\GG_i, i\in\II$ is a successful defense strategy. Let us consider $\ZZ=\{z_i:i\in\II\}\subseteq\UU$. We claim that $\ZZ$ must form a hitting set. Indeed, otherwise let us assume that there exists a $j\in[t]$ such that $\ZZ\cap S_j = \emptyset$. Consider the situation where the attacker attacks voter groups $\GG_i$ for every $i\in[n]\setminus\II$. We observe that $\DD_{\cup_{i\in\II}\GG_i}(y,x_j)=0$ and hence the candidate $y$ is not the Condorcet winner. This contradicts our assumption that defending all the voter groups $\GG_i, i\in\II$ is a successful defense strategy. Hence \ZZ forms a hitting set and thus the \HS instance is a \YES instance.}
\end{proof}}

In the proof of \Cref{thm:hard_td_condorcet}, we observe that the reduced instance of \OD viewed as an instance of the \OA problem is a \NO instance if and only if the \kSum instance is a \YES instance. Hence, the same reduction as in the proof of \Cref{thm:hard_td_condorcet} gives us the following result for the \OA problem.

\begin{corollary}
 The \OA problem for the Condorcet voting rule is \WTH parameterized by $k_d$.\label{cor:hard_ta_condorcet}
\end{corollary}
}

We now show that the \OD problem for scoring rules is \WTH parameterized by $k_a$ also by exhibiting a parameter preserving reduction from a problem closely related to \HS{}, which is \SC problem parameterized by the solution size.\longversion{ The \SC problem is defined as follows.} This is a \WTC{} problem~\cite{downey1999parameterized}. We now \longversion{present}\shortversion{state} our \WT-hardness proof for the \OD problem for scoring rules\shortversion{ and the Condorcet voting rule,} parameterized by $k_a$, by a reduction from \SC{}.

\longversion{
\begin{definition}[\SC]
 Given an universe \UU, a set $\SS = \{S_i: i\in[t]\}$ of subsets of \UU, and a non-negative integer $k$ which is at most $t$, does there exists an index set $\II\subset[t]$ with $|\II|=k$ such that $\bigcup_{i\in\II} S_i = \UU$. We denote an arbitrary instance of \SC by $(\UU, \SS, k)$.
\end{definition}}

\begin{theorem}\shortversion{[$\star$]}\label{thm:hard_ta_sc_con}
 The \OD problem for every scoring rule and Condorcet rule is \WTH parameterized by $k_a$.
\end{theorem}

\longversion{
\begin{theorem}\shortversion{[$\star$]}\label{thm:hard_ta}
 The \OD problem for every scoring rule is \WTH parameterized by $k_a$.
\end{theorem}
\longversion{
\begin{proof}
 Let $(\UU, \SS = \{S_j: j\in[t]\}, k)$ be an arbitrary instance of \SC. Let $\UU=\{z_i:i\in[n]\}$. We assume that $k>3$ since otherwise the \SC instance is polynomial time solvable. For $i\in[n]$, let $f_i$ be the number of $j\in[t]$ such that $z_i\in S_j$; that is, $f_i = |\{j\in[t]: z_i\in S_j\}|$. We assume, without loss of generality, that for every $i\in[n]$, $t-f_i-k>3k$ by adding at most $9t$ empty sets in \SS. We construct the following instance of the \OD problem for the scoring rule induced by the score vector $\overrightarrow{\alpha}$ rule. The set of candidates $\CC = \{x_i: i\in[n]\} \cup \{y,d\}$. Let $\overrightarrow{\alpha}$ be any normalized score vector of length $n+2$. We have the following voter groups.

 \begin{itemize}[topsep=0pt,itemsep=0pt]
  \item For every $j\in[t]$, we have a voter group $\GG_j$. For every $i\in[n]$ and $j\in[t]$ with $z_i\notin S_j$, we have $2$ copies of $\PP_{x_i}^d$.
  \item We have another voter group \HH where, for every $i\in[n]$, we have $2tn+(2(t-f_i-k)+1)$ copies of $\PP_d^{x_i}$ and $2tn$ copies of $\PP_d^y$.
 \end{itemize}

 We define attacker resource $k_a$ to be $k$ and the defender's resource $k_d$ to be $t-k$. This finishes the description of the \OD instance. We first observe that the score of the candidate $d$ is strictly less than the score of every other candidate. We now observe that the candidate $y$ is the unique winner of the election since the score of the candidate $y$ is $2k-1$ more than the score of the candidate $x_i$ for every $i\in[n]$. We now prove that the \OD instance $(\CC, \cup_{j\in[t]}\GG_j\cup\HH, k_a, k_d)$ is equivalent to the \SC instance $(\UU, \SS, k)$.

 In the forward direction, let us suppose that the \SC instance is a \YES instance. Let $\II\subset[t]$ be such that $|\II|=k$ and $\bigcup_{j\in\II} S_j = \UU$. We claim that the defender's strategy of defending the voter groups $\GG_j$ for every $j\in[t]\setminus\II$ results in a successful defense. To see this, we first observe that, if the attacker attacks the voter group \HH, then the candidate $y$ continues to uniquely win the election irrespective of what other voter groups the attacker attacks. Indeed, since $t-f_i-k>3k$ for every $i\in[n]$, the score of the candidate $x_i$ is strictly less than the score of the candidate $y$ irrespective of what other voter groups the attacker attacks. Since, for every $i\in[n]$ and $j\in[t]$, the score of the candidate $x_i$ is not more than the score of the candidate $y$ in the voter group $\GG_j$, we may assume that the attacker attacks the voter group $\GG_j$ for every $j\in\II$ (since they are the only voter groups unprotected except \HH). Now, since $S_j, j\in\II$ forms a set cover of \UU, after deleting the voter groups $\GG_j, j\in\II$, the score of the candidate $x_i$ increases by at most $2(k-1)$ from the original election for every $i\in[n]$. Hence, after deleting the voter groups $\GG_j, j\in\II$, the score of the candidate $x_i$ is still strictly less than the score of the candidate $y$. Hence the candidate $y$ continues to win and thus the defense is successful. Hence the \OD instance is a \YES instance.

 In the other direction, let us suppose that the \OD instance is a \YES instance. We assume, without loss of generality, that the defender protects exactly $t-k$ voter groups. We argued in the forward direction that we can assume, without loss of generality, that the attacker never attacks the voter group \HH. Hence, we can also assume, without loss of generality, that the defender also does not defend the voter group \HH. Let $\II\subset[t]$ be such that $|\II|=k$ and the defender defends the voter group $\GG_j$ for every $j\in[t]\setminus\II$. We claim that the sets $S_j, j\in\II$ forms a set cover of \UU. Suppose not, then let $z_i$ be an element in \UU which is not covered by $S_j, j\in\II$. We observe that attacking the voter groups $\GG_j$ for every $j\in\II$ increases the score of the candidate $x_i$ by $2k$ which makes the candidate $y$ lose in the resulting election (after deleting the voter groups $\GG_j$ for every $j\in\II$) since the score of $x_i$ is strictly more than the score of $y$. This contradicts our assumption that defending the voter group $\GG_j$ for every $j\in[t]\setminus\II$ is a successful defense strategy. Hence $S_j, j\in\II$ forms a set cover of \UU and thus the \SC instance is a \YES instance.
\end{proof}
}
\longversion{
We now present our \WT-hardness proof for the \OD problem for the Condorcet voting rule parameterized by $k_a$.
}
\begin{theorem}\shortversion{[$\star$]}\label{thm:hard_ta_condorcet} The \OD problem for the Condorcet voting rule is \WTH parameterized by $k_a$.
\end{theorem}
\longversion{
\longversion{\begin{proof}}
\shortversion{\begin{proof_sketch}}
 Let $(\UU, \SS = \{S_j: j\in[t]\}, k)$ be an arbitrary instance of \SC. Let $\UU=\{z_i:i\in[n]\}$. We assume that $k>3$ since otherwise the \SC instance is polynomial time solvable. For $i\in[n]$, let $f_i$ be the number of $j\in[t]$ such that $z_i\in S_j$; that is, $f_i = |\{j\in[t]: z_i\in S_j\}|$. We assume, without loss of generality, that for every $i\in[n]$, $t-f_i-k>3k$ by adding at most $9t$ empty sets in \SS. We construct the following instance of the \OD problem for the Condorcet voting rule. The set of candidates $\CC = \{x_i: i\in[n]\} \cup \{y\}$. We have the following voter groups.

 \begin{itemize}[topsep=0pt,itemsep=0pt]
  \item For every $j\in[t]$, we have a voter group $\GG_j$. For every $i\in[n]$ and $j\in[t]$, we have $\DD_{\GG_j}(y,x_i)=2$ if $z_i\notin S_j$ and $\DD_{\GG_j}(y,x_i)=0$ otherwise. We also have $\DD_{\GG_j}(x_i, x_\el)=0$ for every $j\in[t], i,\el\in[n]$ with $i\ne\el$.
  \item We have another voter group \HH where, for every $i\in[n]$, we have $\DD_\HH(x_i,y) = 2(t-f_i-k)$. We also have $\DD_{\HH}(x_i, x_\el)=0$ for every $i,\el\in[n]$ with $i\ne\el$.
 \end{itemize}

 We define attacker resource $k_a$ to be $k$ and the defender's resource $k_d$ to be $t-k$. This finishes the description of the \OD instance. We first observe that the candidate $y$ is a Condorcet winner of the resulting election. \shortversion{The proof of equivalence of the two instances is similar in spirit of the proof of \Cref{thm:hard_ta}.}\longversion{We now prove that the \OD instance $(\CC, \cup_{j\in[t]}\GG_j\cup\HH, k_a, k_d)$ is equivalent to the \SC instance $(\UU, \SS, k)$.

 In the forward direction, let us suppose that the \SC is a \YES instance. Let $\II\subset[t]$ be such that $|\II|=k$ and $\bigcup_{j\in\II} S_j = \UU$. We claim that the defender's strategy of defending the voter groups $\GG_j$ for every $j\in[t]\setminus\II$ results in a successful defense. To see this, we first observe that, we can assume without loss of generality that the attacker does not attack the voter group \HH since the candidate $y$ loses every pairwise election in \HH. Since, for every $i\in[n]$ and $j\in[t]$, the candidate $y$ does not lose any pairwise election in the voter group $\GG_j$, we may assume that the attacker attacks the voter group $\GG_j$ for every $j\in\II$ (since they are the only voter groups unprotected except \HH). Now, since $S_j, j\in\II$ forms a set cover of \UU, after deleting the voter groups $\GG_j, j\in\II$, we have $\DD_{\cup_{j\in[t]\setminus\II}\GG_i\cup\HH}(y,x_i)\ge 2(t-f_i-k+1)-2(t-f_i-k)=2$ for every $i\in[n]$. Hence, after deleting the voter groups $\GG_j, j\in\II$, the candidate $y$ continues to be the Condorcet winner of the remaining profile. Hence the \OD instance is a \YES instance.

 In the other direction, let us suppose that the \OD instance is a \YES instance. We assume, without loss of generality, that the defender protects exactly $t-k$ voter groups. We argued in the forward direction that we can assume, without loss of generality, that the attacker never attacks the voter group \HH. Hence, we can also assume, without loss of generality, that the defender also does not defend the voter group \HH. Let $\II\subset[t]$ be such that $|\II|=k$ and the defender defends the voter group $\GG_j$ for every $j\in[t]\setminus\II$. We claim that the sets $S_j, j\in\II$ forms a set cover of \UU. Suppose not, then let $z_i$ be an element in \UU which is not covered by $S_j, j\in\II$. We observe that $\DD_{\cup_{j\in[t]\setminus\II}\GG_i\cup\HH}(y,x_i)= 2(t-f_i-k)-2(t-f_i-k)=0$ and thus attacking the voter groups $\GG_j$ for every $j\in\II$ makes the candidate $y$ not the Condorcet winner. This contradicts our assumption that defending the voter group $\GG_j$ for every $j\in[t]\setminus\II$ is a successful defense strategy. Hence $S_j, j\in\II$ forms a set cover of \UU and thus the \SC instance is a \YES instance.}
\longversion{\end{proof}}
\shortversion{\end{proof_sketch}}
}
}

We now show that the \OA problem for the scoring rules is \WOH even parameterized by the combined parameter $k_a$ and $k_d$. Towards that, we exhibit a polynomial parameter transform from the \CL problem parameterized by the size of the clique we are looking for which is known to be \WO-complete. \longversion{The \CL problem is defined as follows.

\begin{definition}[\CL]
 Given a graph \GG and an integer $k$, does there exist a clique in \GG of size $k$? We denote an arbitrary instance of \CL by $(\GG,k)$.
\end{definition}
}

\shortversion{
\begin{theorem}\label{thm:hard_tad_scoring_con}
 The \OA problem for every scoring rule and Condorcet rule is \WOH parameterized by $(k_a,k_d)$.
\end{theorem}
}

\longversion{

\begin{theorem}\label{thm:hard_tad_scoring}
 The \OA problem for every scoring rule is \WOH parameterized by $(k_a,k_d)$.
\end{theorem}

\longversion{\begin{proof}
 Let $(\GG=(\VV, \EE),k)$ be an arbitrary instance of the \CL problem. Let $\VV=\{v_i:i\in[n]\}$ and $\EE=\{e_j:j\in[m]\}$. Let $\overrightarrow{\alpha}$ be any arbitrary normalized score vector of length $m+2$. We construct the following instance of the \OA problem for the scoring rule induced by the score vector $\overrightarrow{\alpha}$. The set of candidates $\CC = \{x_j: j\in[m]\} \cup \{y,d\}$. We have the following voter groups.

 \begin{itemize}[topsep=0pt,itemsep=0pt]
  \item For every $i\in[n]$, we have a voter group $\GG_i$. For every $i\in[n]$, we have $10m$ copies of $\PP_d^{x}$ for every $x\in\CC\setminus\{d\}$ in $\GG_i$. We also have two copies of $\PP^d_{x_j}$ in the voter group $\GG_i$ if the edge $e_j$ is incident on the vertex $v_i$, for every $i\in[m]$ and $j\in[m]$.
  \item We have another voter group \HH. We have one copy of $\PP_d^{x_j}$ for every $j\in[m]$ in \HH.
 \end{itemize}

 We define attacker resource $k_a$ to be $k$ and the defender's resource $k_d$ to be $k-2$. This finishes the description of the \OA instance. Let \QQ be the resulting profile; that it $\QQ=\cup_{i\in[n]}\GG_i\cup\HH$. We first observe that the candidate $y$ is the winner of the resulting election since $s_\QQ(y) = s_\QQ(x_j)+3$ and $s_\QQ(y)>s_\QQ(d)$. This completes a description of the construction. Due to lack of space, we defer the proof of equivalence to a longer version of this manuscript. We now prove that the \OA instance $(\CC, \QQ, k_a, k_d)$ is equivalent to the \CL instance $(\GG, k)$.

In the forward direction, let us assume that $\UU=\{v_i:i\in\II\}\subset\VV$ with $|\II|=k$ forms a clique in \GG. We claim that attacking all the voter groups $\GG_i, i\in\II$ forms a successful attack. Indeed, suppose the defender defends all the voter groups $\GG_i, i\in\II$ except $\GG_\el$ and $\GG_{\el^\pr}$. Let $e_{j^\star}$ be the edge between the vertices $v_\el$ and $v_{\el^\pr}$ in \GG. Let the profile after the attack be $\hat{\GG}$; that is, $\hat{\GG}=\cup_{i\in[n]\setminus\II}\GG_i\cup\GG_{\el}\cup\GG_{\el^\pr}\cup\HH$. Then we have $s_{\hat{\GG}}(y) = s_{\hat{\GG}}(x_{j^\star})-1$ and thus the candidate $y$ does not win after the attack. Hence the \OA instance is \YES instance.

 In the other direction, let the \OA instance be a \YES instance. We first observe that the candidate $d$ performs worse than everyone else in every voter group and thus $d$ can never win. Now we can assume, without loss of generality, that the attacker does not attack the voter group \HH since the candidate $y$ is not receiving more score than any other candidate except $d$ in \HH. Let attacking all the voter groups $\GG_i, i\in\II$ with $|\II|\le k$ is a successful attack. We observe that if $|\II|<k$, then defending any $k-2$ of the groups that are attacked foils the attack -- since the candidate $y$ continues to win even after deleting any one group. Hence we have $|\II|=k$. Let us consider the subset of vertices $\UU=\{v_i:i\in\II\}$. We claim that \UU forms a clique in \GG. Indeed, if not, then let us assume that there exists two indices $\el, \el^\pr\in\II$ such that there is no edge between the vertices $v_\el$ and $v_{\el^\pr}$ in \GG. Let us consider the defender strategy of defending all the voter groups $\GG_i, i\in\II\setminus\{\el, \el^\pr\}$. We observe that the candidate $y$ continues to uniquely receive the highest score among all the candidates and thus $y$ wins uniquely in the resulting election. This contradicts our assumption that attacking all the voter groups $\GG_i, i\in\II$ with $|\II|\le k$ is a successful attack. Hence \UU forms a clique in \GG and thus the \CL instance is a \YES instance.
\end{proof}}
We now show similar result as of \Cref{thm:hard_tad_scoring} for the Condorcet voting rule.

\begin{theorem}\label{thm:hard_tad_condorcet}\shortversion{[$\star$]}
 The \OA problem for the Condorcet voting rule is \WOH parameterized by $(k_a,k_d)$.
\end{theorem}

\longversion{
\begin{proof}

 Let $(\GG=(\VV, \EE),k)$ be an arbitrary instance of the \CL problem. Let $\VV=\{v_i:i\in[n]\}$ and $\EE=\{e_j:j\in[m]\}$. We construct the following instance of the \OA problem for the Condorcet voting rule. The set of candidates $\CC = \{x_j: j\in[m]\} \cup \{y\}$. We have the following voter groups.

 \begin{itemize}[topsep=0pt,itemsep=0pt]
  \item For every $i\in[n]$, we have a voter group $\GG_i$. We have $\DD_{\GG_i}(y,x_j)=2$ if the edge $e_j$ is incident on the vertex $v_i$ and $\DD_{\GG_i}(y,x_j)=0$ if the edge $e_j$ is not incident on the vertex $v_i$, for every $i\in[n]$ and $j\in[m]$. We also have $\DD_{\GG_i}(x_\el, x_j)=0$ for every $i\in[n]$, $j, \el\in[m]$, and $j\ne \el$.
  \item We have another voter group \HH where we have $\DD_\HH(x_j,y)=2$ for every $j\in[m]$ and $\DD_{\HH}(x_\el, x_j)=0$ for every $j, \el\in[m]$ and $j\ne \el$.
 \end{itemize}

 We define attacker resource $k_a$ to be $k$ and the defender's resource $k_d$ to be $k-2$. This finishes the description of the \OA instance. Let \QQ be the resulting profile; that it $\QQ=\cup_{i\in[n]}\GG_i\cup\HH$. We first observe that the candidate $y$ is the Condorcet winner of the resulting election. \shortversion{The proof of equivalence of the two instances is similar in spirit of the proof of \Cref{thm:hard_tad_scoring}.}We now prove that the \OA instance $(\CC, \QQ, k_a, k_d)$ is equivalent to the \CL instance $(\GG, k)$.

 In the forward direction, let us assume that $\UU=\{v_i:i\in\II\}\subset\VV$ with $|\II|=k$ forms a clique in \GG. We claim that attacking all the voter groups $\GG_i, i\in\II$ forms a successful attack. Indeed, suppose the defender defends all the voter groups $\GG_i, i\in\II$ except $\GG_\el$ and $\GG_{\el^\pr}$. Let $e_{j^\star}$ be the edge between the vertices $v_\el$ and $v_{\el^\pr}$ in \GG. Let the profile after the attack be $\hat{\GG}$; that is, $\hat{\GG}=\cup_{i\in[n]\setminus\II}\GG_i\cup\GG_{\el}\cup\GG_{\el^\pr}\cup\HH$. Then we have $\DD_{\hat{\GG}}(y,x_{j^\star}) = 0$ and thus the candidate $y$ is not the unique winner after the attack. Hence the \OA instance is \YES instance.

 In the other direction, let the \OA instance be a \YES instance. We can assume, without loss of generality, that the attacker does not attack the voter group \HH since the candidate $y$ loses every pairwise election in \HH. Let attacking all the voter groups $\GG_i, i\in\II$ with $|\II|\le k$ is a successful attack. We observe that if $|\II|<k$, then defending any $k-2$ of the groups that are attacked foils the attack -- since the candidate $y$ continues to be the Condorcet winner even after deleting any one group. Hence we have $|\II|=k$. Let us consider the subset of vertices $\UU=\{v_i:i\in\II\}$. We claim that \UU forms a clique in \GG. Indeed, if not, then let us assume that there exists two indices $\el, \el^\pr\in\II$ such that there is no edge between the vertices $v_\el$ and $v_{\el^\pr}$ in \GG. Let us consider the defender strategy of defending all the voter groups $\GG_i, i\in\II\setminus\{\el, \el^\pr\}$. We observe that the candidate $y$ continues to be the Condorcet winner in the resulting election. This contradicts our assumption that attacking all the voter groups $\GG_i, i\in\II$ with $|\II|\le k$ is a successful attack. Hence \UU forms a clique in \GG and thus the \CL instance is a \YES instance.
\end{proof}
}
}

Once we have a parameterized algorithm for the \OD problem for the parameter $(k_a,k_d)$, an immediate question is whether there exists a kernel for the \OD problem of size polynomial in $(k_a,k_d)$. We know that the \HS problem does not admit polynomial kernel parameterized by the universe size~\cite{downey1999parameterized}.\longversion{ We observe that the reductions from the \HS problem to the \OD problem in \Cref{thm:hard_td,thm:hard_td_condorcet_od_oa} are polynomial parameter transformations.}\shortversion{ It turns out that the reductions from the \HS problem to the \OD and \OA problems in \Cref{thm:hard_td_od_oa,thm:hard_td_condorcet_od_oa,thm:hard_tad_scoring_con} are polynomial parameter transformations.} Hence we immediately have the following corollary.

\begin{corollary}\label{cor:no_poly_kernel}
 The \OD and \OA problems for the scoring rules and the Condorcet rule do not admit a polynomial kernel parameterized by $(k_a,k_d)$.
\end{corollary}

\section{The FPT Algorithm}

%
%

We complement the negative results of \Cref{obs:hard_td} and \longversion{\Cref{thm:hard_ta}}\shortversion{\Cref{thm:hard_ta_sc_con}} by presenting an FPT algorithm for the \OD problem parameterized by $(k_a, k_d)$. In the absence of a defender, that is when $k_d=0$,\longversion{ Yin et al.}~\cite{YinVAH16} showed that the \OD problem is polynomial time solvable for the plurality voting rule. Their polynomial time algorithm for the \OD problem can easily be extended to any scoring rule. Using this polynomial time algorithm, we design the following $\OO^*(k_a^{k_d})$ time algorithm for the \OD problem for scoring rules. \longversion{This result shows that the \OD problem is fixed parameter tractable with $(k_a,k_d)$ as the parameter.}

\begin{theorem}\label{thm:fpt_ka_m}\shortversion{[$\star$]}
 There is an algorithm for the \OD problem for every scoring rule and the Condorcet voting rule which runs in time $\OO^*(k_a^{k_d})$.
\end{theorem}

\longversion{
\begin{proof}
 Let us prove the result for any scoring rule. The proof for the Condorcet voting rule is exactly similar. Initially we run the attacking algorithm over the $n$ voter groups without any group being protected. If a successful attack exists, the algorithm outputs the $k_a$ groups to be deleted. We recursively branch on $k_a$ cases by protecting one of these $k_a$ groups in each branch and running the attacking algorithm again. In addition, the parameter $k_d$ is also reduced by 1 each time a group is protected. When $k_d$=0, the attacking algorithm is run on all the leaves of the tree and a valid protection strategy exists as long as for at least one of the leaves the attack outputs no i.e. after deploying resources to protect $k_d$ groups the attacker is unable to change the outcome of the election with any strategy. The groups to be protected is determined by traversing the tree that leads to the particular leaf which did not output an attack. Clearly the number of nodes in this tree is bounded by $k_a^{k_d}$. The amount of time taken to find an attack at each node is bounded by $poly(n)$. Hence the running time of this algorithm is bounded by $k_a^{k_d}.poly(n)$.
\end{proof}
}

\section{Experiments}

\begin{figure}[h!]
    \centering
    \includegraphics[scale=0.4]{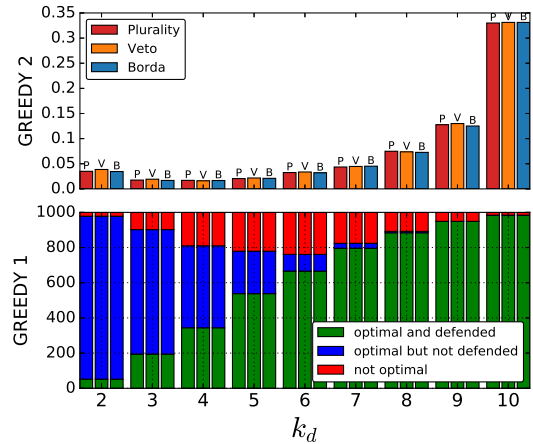}
    \caption{Performances of \gone\ and \gtwo\ for uniform voting profile generation model.}
    \label{fig:Plot1}
\end{figure}

Though the previous sections show that the optimal defending problem is computationally intractable, it is a worst-case result. In practice, elections have voting profiles that are generated from some (possibly known) distribution. In this section, we conduct an empirical study to understand how simple defending strategies perform for two such statistical voter generation models. The defending strategies we consider are variants of a simple greedy policy.

\smallskip \noindent
{\em Defending strategy}: For a given voting profile and a voting rule, the defending strategy finds the winner. Suppose the winner is $a$. The strategy considers $a$ with every other candidate, and for each such pair it creates a sorted list of classes based on the winning margin of votes for $a$ in those classes, and picks the top $k_d$ classes to form a sub-list. 
Now, among all these $(m-1)$ sorted sub-lists, the strategy picks the most frequent $k_d$ classes to protect. We call this version of the strategy \gone.
For certain profiles an optimal attacker (a) may change the outcome by attacking some of the unprotected classes or (b) is unable to change the outcome. If (a) occurs, then there is a possibility that for the value of $k_d$ there does not exist any defense strategy which can guard the election from all possible strategies of the attacker. In that case, \gone\ is optimal and is not optimal otherwise. It is always optimal for case (b). Note that, given a profile and $k_d$ protected classes, it is easy to find if there exists an optimal attack strategy, while it is not so easy to identify whether there does not exist any defending strategy if the \gone\ fails to defend. We find the latter with a brute-force search for this experiment.
A small variant of \gone\ is the following: when \gone\ is unable to defend (which is possible to find out in poly-time), the strategy chooses to protect $k_d$ classes uniformly at random. Call this strategy \gtwo.

\begin{figure}[h!]
    \centering
    \includegraphics[scale=0.4]{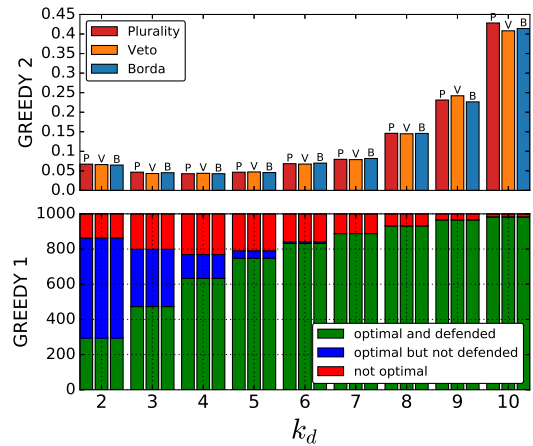}
    \caption{Performances of \gone\ and \gtwo\ for voting profile generation model with two major contesting candidates.}
    \label{fig:Plot2}
\end{figure}

\smallskip \noindent
{\em Voting profile generation}: Fix $m = 5$. 
We generate $1000$ preference profiles over these alternatives for $n = 12000$, where each vote is picked uniformly at random from the set of all possible strict preference orders over $m$ alternatives. The voters are partitioned into $12$ classes containing equal number of voters. We consider three voting rules: plurality, veto, and Borda. The lower plot in \Cref{fig:Plot1} shows the number of profiles which belongs to the three categories: (i) \gone\ defends (is optimal), (ii) \gone\ cannot defend but no defending strategy exists (is optimal), (iii) \gone\ cannot defend but defending strategy exists (not optimal). The x-axis shows different values of $k_d$ and we fix $k_a = 12 - k_d$.

The upper plot of \Cref{fig:Plot1} shows the fraction of the profiles successfully defended by \gtwo\ where \gone\ is not optimal (i.e., cannot defend but defending strategy exists) when \gtwo\ uniformly at random picks $k_d$ classes $100$ times. These fractions therefore serves as an empirical probability of successful defense of \gtwo\ given \gone\ is not optimal.

In an election where the primary contest happens between two major candidates, even though there are more candidates present, the generation model may be a little different. We also consider another generation model that generates $40\%$ profiles having a fixed alternative $a$ on top and the strict order of the $(m-1)$ alternatives is picked uniformly at random, a similar $40\%$ profiles with some other alternative $b$ on top, and the remaining $20\%$ preferences are picked uniformly at random from the set of all possible strict preference orders. Similar experiments are run on this generation model and results are shown in \Cref{fig:Plot2}.

The results show that even though optimal defense is a hard problem, a simple strategy like greedy achieves more than $70\%$ optimality. From the rest $30\%$ non-optimal cases, the variant \gtwo\ is capable of salvaging it into optimal with probability almost $5\%$ for uniform generation model and above $5\%$ for two-major contestant generation model for $k_d = k_a = 6$. This empirically hints at a possibility that defending real elections may not be too difficult.

\section{Conclusion}

We have considered the \OD{} problem from a primarily parameterized perspective for scoring rules and the Condorcet voting rule. We showed hardness in the number of candidates, the number of resources for the defender or the attacker. On the other hand, we show tractability for the combined parameter $(k_a, k_d)$. We also introduced the \OA{} problem, which is hard even for the combined parameter $(k_a, k_d)$, and also showed the hardness for a constant number of candidates. 
Even though the \OD{} problem is hard, empirically we show that relatively simple mechanisms ensure good defending performance for reasonable voting profiles.


\bibliographystyle{plain}
\bibliography{references}

\end{document}